\documentclass[nohyperref]{article}

\usepackage[utf8]{inputenc}
\usepackage[T1]{fontenc}
\usepackage[fontsize=12]{scrextend}
\usepackage[english]{babel}

\usepackage[sf,mono=false]{libertine} 
\usepackage{bold-extra} % for bold sscshape

\usepackage{subfiles}
\usepackage{makeidx}
\usepackage{lipsum}

\usepackage{geometry}
\usepackage{graphicx}
\newlength\myinner
\newlength\myouter
\usepackage{tikz}
\usetikzlibrary[patterns] % ConTeXt when using TikZ
\usetikzlibrary{arrows.meta,backgrounds,fit,positioning,petri}
\usepackage{tikz-cd} % for commutative diagrams

\usepackage{color}
\definecolor{aqua}{rgb}{0.0, 1.0, 1.0}
\definecolor{aquamarine}{rgb}{0.5, 1.0, 0.83}
\definecolor{alizarin}{rgb}{0.82, 0.1, 0.26}
\definecolor{carmine}{rgb}{0.59, 0.0, 0.09}
\definecolor{tred}{rgb}{0.8, 0.0, 0.13}
\definecolor{blue(pigment)}{rgb}{0.2, 0.2, 0.6}
\definecolor{darkgreen}{rgb}{0.0, 0.39, 0.0}
\definecolor{Crimson}{HTML}{DC143C}
\definecolor{Maroon}{HTML}{800000}
\definecolor{MidnightBlue}{HTML}{191970}
\definecolor{Peru}{HTML}{CD853F}
\definecolor{Teal}{HTML}{008080}
\definecolor{FloralWhite}{HTML}{FFFAF0}
\definecolor{LightYellow}{HTML}{FFFFE0}
\definecolor{Burgundy}{HTML}{9C001A}

\usepackage{hyperref}
\hypersetup
{
	colorlinks=true,
	linkcolor=blue(pigment),
	filecolor=magenta,			
	urlcolor=Burgundy,
	citecolor=darkgreen
}

\usepackage[inline, shortlabels]{enumitem}
\usepackage{comment}

\usepackage{algorithm}
\usepackage{algpseudocode}

\usepackage{dsfont}
\usepackage{mathrsfs, mathtools}
\usepackage{amsthm, amsmath, amssymb}
\usepackage{stmaryrd}
\usepackage{pifont}
\theoremstyle{plain}
\newtheorem{theorem}{Theorem}

\newtheorem{lemma}{Lemma}

\theoremstyle{definition}

\theoremstyle{remark}
\newtheorem{remark}{Remark}

\usepackage[capitalize,noabbrev,nameinlink]{cleveref}

\newcommand\E{\mathbb{E}}

\renewcommand{\P}{\mathbb{P}}

\newcommand{\parens}[1]{\left(#1\right)}
\newcommand{\brackets}[1]{\left[#1\right]}

\newcommand{\set}[1]{\left\{#1\right\}}
\newcommand{\abs}[1]{\left|#1\right|}

\newcommand{\tuple}[1]{\left\langle#1\right\rangle}
\newcommand{\dotProduct}[2]{\left\langle#1,#2\right\rangle}

\reversemarginpar

\newcommand{\Ac}{{\mathcal A}}

\newcommand{\Sc}{{\mathcal S}}

\newcommand{\Xc}{{\mathcal X}}

\title{\bf When do discounted-optimal policies also optimize the gain?}
\author{Victor Boone\footnote{victor.boone@univ-grenoble-alpes.fr}\\
  Univ. Grenoble-Alpes, Inria, CNRS, LIG, 38300 Grenoble, France}

\def\sp#1{{\rm sp}(#1)}
\def\Pr{\P}

\begin{document}

\maketitle

\begin{abstract}
    In this technical note, we establish an upper-bound on the threshold on the discount factor starting from which all discounted-optimal deterministic policies are gain-optimal, that we prove to be tight on an example. 
    To address computability issues of that theoretical threshold, we provide a weaker bound which is tractable on ergodic MDPs in polynomial time. 
\end{abstract}

\section{Concepts and main result}

We consider Markov decision processes (MDPs) with finitely many states and actions, whose sets are respectively denoted $\Xc$ and $\Ac = \bigcup_{x \in \Xc} \Ac(x)$. 
A MDP is given by a tuple $M = \tuple{\Xc, \Ac, p, q}$ where $p$ and $q$ are respectively the transition kernel and reward distributions.
% Specifically, $p(y|x,a)$ denotes is the probability to end in $y \in \Xc$ from choosing action $a \in \Ac(x)$ from state $x \in \Xc$; and $q(x,a)$ is the distribution of the reward a decision maker gets by choosing action $a \in \Ac(x)$ from state $x \in \Xc$. 
The mean reward associated to $(x, a)$ is denoted $r(x,a)$, viz., $r(x, a) := \E_{R \sim q(x,a)}[R]$. 

A \emph{policy} $\pi \in \Pi$ is any deterministic stationary decision rule.
% , i.e., an element of $\Pi := \prod_{x \in \Xc} \Ac(x)$. 
Upon iterating a policy $\pi$ on $M$ starting from $x \in \Xc$, we obtain a sequence of states, actions, and rewards $\set{(X_t, A_t, R_t) : t \ge 0}$ whose probability measure will be denoted $\Pr^\pi_x(-)$ and expectation operator $\E^\pi_x[-]$.
We will write $Z_t$ for the pair $(X_t, A_t)$. 
The iterates of a policy also define a Markov reward process (MRP) $\set{(X_t, R_t)}$ whose transition kernel will be denoted $P^\pi$ and mean reward vector $r^\pi$, i.e., $r^\pi(x) := r(x, \pi(x))$.
To each policy are associated various notions of scores:
\begin{itemize}
    \item the finite-horizon score ${\rm J}_T^\pi(x) := \E_x^\pi[\sum_{t=0}^{T-1} r(Z_t)]$; 
    \item the $\beta$-discounted score ${\rm V}_\beta^\pi(x) := \E_x^\pi[\sum_{t=0}^\infty r(Z_t) \beta^{t}]$ for $\beta \in [0, 1)$;
    \item the gain $g^\pi(x) := \lim_{T \to \infty} \E_x^\pi[\frac 1T \sum_{t=0}^{T-1} r(Z_t)]$;
    \item the bias $h^\pi(x) := \lim_{T \to \infty} \E_x^\pi[\sum_{t=0}^\infty (r(Z_t) - g^\pi(X_t))]$, or the Ces\'aro-limit when the limit doesn't exist.
\end{itemize}

Note that those quantities depend of the underlying MDP $M$ -- this $M$ will sometimes be added to notations to avoid ambiguities. 

A policy that achieves maximal $\beta$-discounted score from all state is said \emph{$\beta$-discounted optimal} and we write $\pi \in \Pi^*_\beta$.
A policy that achieves maximal gain from all state is said \emph{gain-optimal} and we write $\pi \in \Pi_{-1}^*$; if in addition it achieves maximal bias from all state, it is said \emph{bias-optimal} and we write $\pi \in \Pi_0^*$.
All these classes are non-empty \cite{puterman_markov_1994}. 
It is known \cite{blackwell_discrete_1962} that when $\beta \uparrow 1$, $\Pi_\beta^*$ eventually converges to a single class of policies known as \emph{Blackwell-optimal policies}, which is a non-trivial subset of bias-optimal policies \cite{puterman_markov_1994}. 

There are few explicit bounds on how large $\beta$ needs to be so that all $\beta$-discounted optimal policies are Blackwell-optimal \cite{grandclement2023reducing}.
For  a less demanding problem, and to the best of our knowledge, the question of finding a threshold on the discount factor $\beta_{-1}$ that  guarantees that every $\beta$-discounted optimal policy is  gain-optimal ({\it i.e.} in $\Pi^*_{-1}$) for all $\beta > \beta_{-1}$ has not been addressed so far.
This threshold is formally given by:

\begin{equation}
    \beta_{-1} := 
    \inf \set{
        \beta^* \ge 1 
        : \forall \beta \in (\beta^*, 1], ~\Pi_\beta^* \subseteq \Pi_{-1}^*
    }.
\end{equation}
\cref{theorem:main} establishes an upper-bound on $\beta_{-1}$.

If $u$ is a vector (e.g., $g^\pi$, $h^\pi, \ldots$), its \emph{span} is $\sp{u} := \max_x u(x) - \min_x u(x)$. 
We denote $g^*$ and $h^*$ the respective optimal gain and bias vectors, equal to  $g^{\pi^*}$ and $h^{\pi^*}$ respectively, where $\pi^*$ is any bias-optimal policy. 

\begin{theorem}
    \label{theorem:main}
    For all MDP with finitely many states $\Xc$ and actions, 
    \begin{equation}
        \label{equation:main}
        \beta_{-1} \le 1 - \inf \set{ \frac{g^*(x) - g^\pi(x)}{\sp{h^*} + \sp{h^\pi}} : x \in \Xc \text{ and $\pi \in \Pi$ s.t.~}g^\pi(x) < g^*(x)}.
    \end{equation}%
\end{theorem}%
This result underlines a trade-off to be found between how close to gain-optimal $\pi$ is and how large its bias tends to be. 
This result follows from relatively folklore identities linking the gain and the discounted score. 
A complete proof is provided later in this note. 

The bound \eqref{equation:main} is tight. 
Consider the MDP with deterministic transitions pictured in \cref{figure:dmdp}.
There is a single action from states $2, 3$ and two from $1$ (left or right). 
The associated mean rewards are represented as arc weights. 

\begin{figure}[h]
    \centering
    \begin{tikzpicture}
        \node[circle, draw] (1) at ( 0, 0) {$0$};
        \node[circle, draw] (2) at ( 3, 0) {$1$};
        \node[circle, draw] (3) at (-3, 0) {$2$};
        \draw[->, >=stealth] (1) to node[midway, above] {$1$} (2);
        \draw[->, >=stealth] (1) to node[midway, above] {$1 + \epsilon_h - \epsilon_g$} (3);
        \draw[->, >=stealth] (3) to[in=90,out=180-45] (-4, 0) node[left] {$1 - \epsilon_g$} to[in=180+45,out=-90] (3);
        \draw[->, >=stealth] (2) to[in=90,out=+45] (4, 0) node[right] {$1$} to[in=-45,out=-90] (2);
    \end{tikzpicture}
    \caption{
        \label{figure:dmdp}
        A deterministic MDP achieving the bound of \cref{theorem:main} ($\epsilon_g, \epsilon_h > 0$). 
    }
\end{figure}
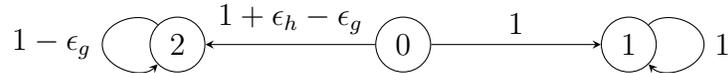

The policy going to $i \in \set{1, 2}$ from $0$ is denoted $\pi_i$.
Only $\pi_1$ is gain-optimal.
\begin{equation*}
    \begin{array}{llll}
        \pi_1: & {\rm V}_\beta^{\pi_1}(0) = (1 - \beta)^{-1}, & g^{\pi_1}(0) = 1, & h^{\pi_1} = (0, 0, 0); \\
        \pi_2: & {\rm V}_\beta^{\pi_2}(0) = (1 - \epsilon_g)(1 - \beta)^{-1} + \epsilon_h, & g^{\pi_1}(0) = 1 - \epsilon_g, & h^{\pi_2} = (0, \epsilon_h, 0).
    \end{array}
\end{equation*}
From the values above, we derive that 
$$
    {\rm V}_\beta^{\pi_1}(0) > {\rm V}_\beta^{\pi_2}(0)
    \iff 
    \beta > 1 - \frac{\epsilon_g}{\epsilon_h} = 1 - \frac{g^{\pi_1}(0) - g^{\pi_2}(0)}{\sp{h^{\pi_1}} + \sp{h^{\pi_2}}}.
$$

\section{A tractable lower bound for ergodic MDPs}

The inf-bound in \cref{theorem:main} is hard to compute because one needs to check all suboptimal policies. 
Our second result is a weaker version of \eqref{equation:main} for \emph{ergodic} MDPs that can be computed in polynomial time.
Recall that a MDP is \emph{ergodic} if $P^\pi$ is ergodic for all $\pi$ -- this condition is not easy to check in general, and is usually decided by the problem that the MDP models. 

\begin{theorem}
    \label{theorem:ergodic}
    For all ergodic MDP with finitely many states $\Xc$ and actions, 
    \begin{equation}
        \beta_{-1} \le 1 - \frac{\Delta_g}{2\sp{r}\overline{D}}
    \end{equation}
    where $\Delta_g := \inf\set{g^*(x) - g^\pi(x) : x \in \Xc \text{ and } g^\pi(x) < g^*(x)} > 0$ is the \emph{gain-gap} of the MDP 
    and $\overline{D} := \max_\pi \max_{x \ne y} \E_x^\pi[\tau_y]\footnote{$\tau_y$ is the reaching time of $y$, starting from the initial state: $\tau_y = \inf\{ t\geq 0 : X_t = y\}$.} < \infty$ is the \emph{worst diameter}.
\end{theorem}

We later show that $\Delta_g$ and $\overline{D}$ are both computable in polynomial time.

\begin{proof}
    We only have to show that for all policy $\pi$, $\sp{h^\pi} \le \overline{D}\sp{r}$.
    This result is variation on \cite[Theorem~4]{bartlett_regal_2009} that we prove using a different technique. 
    Consider the Markov chain $\set{X_t : t \ge 1}$ induced by the iterations of $\pi$. 
    Let $x, y \in \Xc$ such that $\sp{h^\pi} = h^\pi(x) - h^\pi(y)$. 
    Because $r^\pi - g^\pi = (I - P^\pi) h^\pi$, the quantity
    $$
        \parens{r^\pi(X_t) - g^\pi(X_t)} - \parens{h^\pi(X_t) - h^\pi(X_{t+1})}
    $$
    is a martingale difference sequence. 
    Its differences are a.s.~bounded by $\sp{r} + \sp{h^\pi} < \infty$ and $\tau_y$ is an a.s.~finite stopping time. 
    Hence:
    \begin{align*}
        h(x) - h(y) = \E^\pi_x \brackets{h^\pi(X_0) - h^\pi(X_{\tau_y})}
        & = \E_x^\pi\brackets{\sum\nolimits_{t=0}^{\tau_y-1} (r^\pi(X_t) - g^\pi(X_t))} \\
        & \le \E_x^\pi[\tau_y] \sp{r} \\
        & \le \overline{D}\sp{r}. \qedhere
    \end{align*}
\end{proof}

\begin{remark}
    We see from the definition that $\overline{D}< \infty$ if and only if $M$ is ergodic. 
    If $M$ is not ergodic, then the bound of \cref{theorem:ergodic} is not informative. 
\end{remark}

\subsection{Computation of $\Delta_g$}

Define, for $(x, a)$ a state-action pair, 
\begin{equation}
    \Delta^*(x, a) := h^*(x) - \brackets{r(x, a) - g^*(x) + \dotProduct{p(x,a)}{h^*}} 
    % \ge 0
\end{equation}
the suboptimality gap of $(x, a)$, which is non-negative.
For $\pi \in \Pi$, $\mu^\pi_x$ denotes the (empirical) invariant measure of $\pi$ achieved by iterating $\pi$ starting from $x \in \Xc$. 

\begin{lemma}
    \label{lemma:bellman gap}
    For all policy $\pi$ and all $x \in \Sc$, $g^\pi(x) \le g^*(x) - \sum_{y \in \Xc} \mu^\pi_x(y) \Delta^*(y, \pi(y))$.
\end{lemma}
\begin{proof}
    This result can either be established algebraically or using the martingale technique used in \cref{theorem:ergodic}.
    We go for the algebraic proof here. 
    Denote $\Delta^\pi(x) := \Delta^*(x, \pi(x))$. 
    By definition of $\Delta^*$, we have $r^\pi = g^* + (I - P^\pi) h^* + \Delta^\pi$. 
    Multiplying by $(P^\pi)^t$ and summing over $t$, we obtain:
    \begin{align*}
        \sum\nolimits_{t=0}^{T-1} (P^\pi)^t r^\pi 
        & = 
        \sum\nolimits_{t=0}^{T-1} (P^\pi)^t g^*
        + \parens{I - (P^\pi)^T}h^*
        + \sum\nolimits_{t=0}^{T-1} (P^\pi)^t \Delta^\pi
        \\
        & \le 
        T g^*
        + \parens{I - (P^\pi)^T}h^*
        + \sum\nolimits_{t=0}^{T-1} (P^\pi)^t \Delta^\pi.
    \end{align*}
    Dividing by $T$ and making $T$ go to infinity, we obtain:
    \begin{equation}
        \notag
        g^\pi \le g^* + \parens{\lim_{T \to \infty} \frac 1T \sum\nolimits_{t=0}^{T-1} (P^\pi)^t} \Delta^\pi
    \end{equation}
    whose $x$-th line readily provides the result.
\end{proof}

Following \cref{lemma:bellman gap}, we deduce that for any policy $\pi$ such that $g^\pi(x) < g^*(x)$, there must be $y \in \Xc$ such that $\mu_x^\pi(y) \Delta^*(y, \pi(y)) > 0$. 
More precisely, since  $\mu_x^\pi(y) >0$ for all $\pi$ and all $y$ in the ergodic case,   a policy is gain-suboptimal  if and only if it uses a suboptimal action at some point, i.e., an action such that $\Delta^*(x, a) > 0$. 
Because, given a MDP $M$, the computation of $g^*(M)$ is polynomial time, we deduce that $\Delta_g$ is also  computed in polynomial time with the following procedure:

\begin{algorithm}[h]
    \caption{Computation of $\Delta_g$ for a MDP $M$}
    \begin{algorithmic}[1]
        \State {Construct $\set{M_{xa} : (x,a) \in \Sc \times \Ac}$ where $M_{xa}$ is the copy of $M$ whose only available action from $x$ is $a$; }
        \State {Compute $\set{g^*(M_{xa}) : (x, a) \in \Sc \times \Ac}$ and $g^*(M)$; }
        \State{\Return $\min\set{g^*(M) - g^*(M_{xa}) : g^*(M_{xa}) < g^*(M)}$.}
    \end{algorithmic}
\end{algorithm}

% \begin{remark}
%     This algorithm doesn't require $M$ to be ergodic. 
%     In fact, this quantity can be computed in polytime regardless of the MDP's properties, as soon as the computation of $g^*$ is polytime. 
% \end{remark}

\subsection{Computation of $\overline{D}$}

The computation of $\overline{D}$ follows the same idea than $\Delta_g$. 
Denote $M_{y}$ the copy of $M$ where (1) $y$ is zero-reward absorbing state and (2) all rewards, except from $y$, are set to $1$. 
Because $y$ is recurrent under every policy $\pi$ (on $M$), the iterates of $\pi$ are eventually stationary to $y$ on $M_y$.
In particular, $g^\pi(M_y) = 0$ for all $\pi$ and $y$. 
We deduce that $h^*(x; M_y) = \max_\pi \E^{\pi, M}_x[\tau_y]$, hence:

$$
    \max_x h^*(x; M_y) = \max_\pi \max_x \E_x^{\pi, M}[\tau_y].
$$

But is the computation of $h^*(M_y)$ polytime?
In general, Bellman's equations are not enough to compute $h^*$.
We show that they are sufficient for $M_y$.

\begin{lemma}
    Write $M_y = \tuple{\Sc, \Ac, p_{M_y}, q_{M_y}}$. 
    Let $\pi$ any policy that satisfies the Bellman equation:
    $$
        \forall x \in \Xc, 
        \quad
        g^\pi(x; M_y) + h^\pi(x; M_y) 
        = \max_{a \in \Ac(x)} \set{ r(x, a; M_y) + \dotProduct{p_{M_y}(x, a)}{h^\pi(M_y)}}.
    $$
    Then $\pi$ is bias-optimal on $M_y$. 
\end{lemma}
\begin{proof}
    This result is a special case of a much more general result that says that, if all policies have the same recurrent states, Bellman equations automatically guarantee bias-optimality.
    We provide an ad-hoc proof for the special case of $M_y$ for self-containedness.
    For short, denote $g'^\pi$, $h'^\pi$, $r'^\pi$ and $P'^\pi$ the quantities related to $\pi$ on $M_y$. 
    We know that $g'^\pi = 0$ automatically (all policies are gain-optimal on $M_y$). 
    Let $\pi^*$ bias-optimal on $M_y$.
    By applying Bellman's equation iteratively, we get:
    \begin{align*}
        h'^\pi 
        & \ge r'^{\pi^*} + P'^{\pi^*} h'^\pi \\
        & \ge r'^{\pi^*} + P'^{\pi^*} \parens{r'^{\pi^*} + P'^{\pi^*} h'^\pi} \\
        & ~\!~\vdots \\
        & \ge \sum\nolimits_{t=0}^{T-1} P'^{\pi^*} r'^{\pi^*} + (P'^{\pi^*})^T h'^\pi.
    \end{align*}
    We know that the only recurrent state of $\pi^*$ is $y$, so when $T \to \infty$, $(P'^{\pi^*})^T h'^\pi \to h'^\pi(y) = 0$.
    So overall, when $T \to \infty$, we get $h'^\pi \ge h^*$. 
    So $h'^\pi = h^*$.
\end{proof}

In the end, $\overline{D}$ is computed in polynomial type as follows.

\begin{algorithm}[h]
    \caption{Computation of $\overline{D}$ for ergodic MDPs.}
    \begin{algorithmic}[1]
        \State{Construct $\set{M_y: y \in \Sc}$ where $M_y$ is the $y$-absorbing copy of $M$ with $0$ reward on $y$ and $1$ reward everywhere else;}
        \State{For each $y$, compute $\pi_y$ a policy satisfying the Bellman's equations on $M_y$;}
        \State{\Return $\max_y \max_x h^{\pi_y}(x; M_y)$.}
    \end{algorithmic}
\end{algorithm}

\section{Proof of \cref{theorem:main}}

\begin{lemma}
    \label{lemma:finite horizon}
    Every policy $\pi$ satisfies:
    $
        g^\pi(x) - \frac 1T \sp{h^\pi}
        \le 
        \frac 1T {\rm J}^\pi_T(x)
        \le 
        g^\pi(x) + \frac 1T \sp{h^\pi}.
    $
\end{lemma}
\begin{proof}[Proof of \cref{lemma:finite horizon}]
    Bias, gain and reward vectors are linked by the following Poisson equation: $r^\pi = g^\pi + (I - P^\pi)h^\pi$. 
    Multiplying by $(P^\pi)^t$ and summing up, we obtain:
    \begin{align*}
        {\rm J}_T^\pi := \sum_{t=0}^{T-1} (P^\pi)^t r^\pi
        & = \sum_{t=0}^{T-1} (P^\pi)^t g^\pi + \sum_{t=0}^{T-1} \parens{(P^\pi)^t - (P^\pi)^{t+1}} h^\pi \\
        & = T g^\pi + (I - (P^\pi)^T) h^\pi
    \end{align*}
    where the last equality is obtained using $P^\pi g^\pi = g^\pi$. 
    Now, because $I - (P^{\pi})^T$ is a difference of line-stochastic matrices, all entries of $(I - (P^\pi))^T h^\pi$ are upper bounded by $\sp{h^\pi}$ in absolute value. 
    If $e$  denotes the vector whose entries are all $1$s, we thus obtain
    $$
        T g^\pi - \sp{h^\pi} e 
        \le {\rm J}_T^\pi
        \le T g^\pi + \sp{h^\pi} e
    $$
    which is the claimed result. 
\end{proof}

\begin{lemma}
    \label{lemma:discounted}
    For all policy $\pi$ and discount factor $\beta \in [0, 1)$, 
    $
        \abs{{\rm V}_\beta^\pi(x) - \frac{g^\pi(x)}{1-\beta}}
        \le 
        \sp{h^\pi}.
    $
\end{lemma}
\begin{proof}[Proof of \cref{lemma:discounted}]
    Denote $r_t := \E_x^\pi[r(Z_t)]$ the $t$-th expected reward.
    Let $F_\beta^\pi(x):= (1 - \beta) \sum_{t=0}^\infty \beta^{t} r_t$ the normalized $\beta$-discounted reward of $\pi$ starting from $x$. 
    Then for $\abs{\beta} < 1$, 
    \begin{align*}
        (1 - \beta)^{-2} F_\beta^\pi(x)
        & = (1 - \beta)^{-1} \sum_{t=0}^\infty r_t \beta^{t} 
        \\
        & = \parens{\sum_{t=0}^\infty \beta^{t}} \parens{\sum_{t=0}^\infty r_t \beta^{t}}
        \\ 
        & = \sum_{t=0}^\infty {\rm J}_{t+1}^\pi(x) \beta^{t}.
    \end{align*}
    So $F_\beta^\pi(x) = (1-\beta)^2 \sum_{t=0}^\infty {\rm J}_{t+1}^\pi (x)\beta^{t}$.
    Using $1 = (1 - \beta)^2 \sum_{t=0}^\infty  (t+1) \beta^{t}$,
    we also have  $g^\pi = (1 - \beta)^2 \sum_{t=0}^\infty g^\pi (t+1) \beta^{t}$, and 
    \begin{align*}
        \abs{
            F_\beta^\pi(x) - g^\pi(x)
        }
        & \le 
        (1 - \beta)^2
        \sum_{t=0}^\infty
          \abs{\frac{1}{t+1} {\rm J}_{t+1}^\pi(x) - g^\pi(x)} (t+1)\beta^{t}\\
      & \leq (1 - \beta)^2 \sum_{t=0}^\infty  \sp{h^\pi} \beta^{t}\\
       &= (1-\beta) \sp{h^\pi},
    \end{align*}
    where the second inequality is obtained by applying \cref{lemma:finite horizon} for $J_{t+1}^\pi$.
\end{proof}

\begin{proof}[Proof of \cref{theorem:main}]
    Let $\beta^*$ the right-hand side of \eqref{equation:main} and let $\beta > \beta^*$.
    Let $\pi^*$ a bias-optimal policy and $\pi$ such that $g^\pi(x) < g^*(x)$ for some $x \in \Xc$. 
    From \cref{lemma:discounted} follows that
    $$
        {\rm V}_\beta^{\pi^*}(x)
        \ge
        {\rm V}_\beta^{\pi}(x)
        +
        \frac{g^*(x) - g^\pi(x)}{1 - \beta} - \sp{h^*} - \sp{h^\pi}.
    $$
    A sufficient condition for ${\rm V}_\beta^{\pi^*}(x) > {\rm V}_{\beta}^\pi(x)$ is thus
    $$
        \frac{g^*(x) - g^\pi(x)}{1 - \beta} - \sp{h^*} - \sp{h^\pi}
        > 0.
    $$
    Solving the above inequality in $\beta$, this sufficient condition is equivalent to:
    $$
        \beta > 1 - \frac{g^*(x) - g^\pi(x)}{\sp{h^*} + \sp{h^\pi}}.
    $$
    This holds by choice of $\beta^*$.
    So ${\rm V}_\beta^{\pi^*}(x) > {\rm V}_\beta^{\pi}(x)$ and $\pi$ is not $\beta$-discounted optimal. 
\end{proof}

\bibliographystyle{apalike}
\bibliography{bibliography}

\end{document}